\documentclass[12pt, a4paper]{article}
\usepackage{amssymb}
\usepackage{amsthm}
\usepackage{amsfonts}
\usepackage[english]{babel}
\usepackage{graphicx}
\usepackage{appendix}
\usepackage{amsmath}

\newtheorem{theorem}{Theorem}[section]

\newtheorem{definition}[theorem]{Definition}
\newtheorem{proof*}{proof }

\newcommand{\bean}{\begin{eqnarray*}}
\newcommand{\eean}{\end{eqnarray*}}
\newcommand{\ba}{\begin{array}}
\newcommand{\ea}{\end{array}}
\newcommand{\be}{\begin{equation}}
\newcommand{\ee}{\end{equation}}
\newcommand{\bea}{\begin{eqnarray}}
\newcommand{\eea}{\end{eqnarray}}
\newcommand{\pa}{\partial}

\newcommand{\no}{\nonumber}
\newcommand{\om}{\omega}

\begin{document}

\title
{Constructions of Totally Non-Negative Pfaffian}
\author{
 Jen-Hsu Chang \\Graduate School of National Defense, \\
 National Defense University, \\
 Tauyuan City,  335009, Taiwan }

\date{}

\maketitle
\begin{abstract}
The totally non-negative pfaffian (TNNP) is define for  a skew-symmetric matrix such that all the sub-pfaffians are
non-negative. It appears in the pfaffian structure of $\tau$-function for the non-singular web solitons of the BKP equation . One constructs  TNNP  using the Perfect matching, chord diagram and the Dyck  paths enumeration.   Its tridiagonal form is also investigated. 

\end{abstract}
Keywords: BKP equation, Totally non-negative pfaffian, Perfect matching,  Tridiagonal Form \\
2020 Mathematics Subject Classification: 5C90; 15A15; 15A23; 37K10
\newpage

\section{Introduction} 
The pfaffian structure often appears in the soliton theory via the $\tau$- function theory, such as the  
$B$-Type Kadomtsev-Petviashvili  Equation (BKP) equation \cite{da, dj, de, lw, rh, kv, ni, os}  or the (2+1) Sawada-Kotera equation \cite{ji} and  the Veselov-Novikov equation defined on integrable 2D Schrodinger operator \cite{ba,df, jh1, gr, kr}. Also, the pfaffian technique is used to construct the non-singular rational solutions such that the solutions decay to zero in each  direction  on the plane \cite{gn}. The reason for pfaffian appearing in soliton theory comes from the Hirota equations, whose solutions boil down to some pfaffian identities when $\tau$-function is defined by pfaffian \cite{hi}. Recently, the TNNP \cite{jh2, yt} is used to construct non-singular line solitons and their resonance is investigated for the BKP equation. The BKP equation or the (2+1) Sawada-Kotera equation \cite{ji} is 
\be 
(9 \phi_t-5\phi _{xxy}+\phi_{xxxxx}-15 \phi_x \phi_y+15 \phi_{x}\phi_{xxx}+15\phi_x^3)_x-5 \phi_{yy}=0. \label{bkp}
\ee
This equation arises from the B-type Lie algebras and it is known that  the KP equation arises from the A-type algebras. The BKP hierarchy is sub-hierarchy of KP hierarchy. The Schur Q functions  are used to construct rational solutions \cite{ro, or} and the web soliton solutions from a  given TNNP obtained by a skew Schur Q function \cite {jh2}. The bilinear form of the BKP equation (\ref{bkp}) is \cite{dj}
\be 
(D_x^6-5D_x^3 D_y -5 D_y^2+9 D_x D_t) \tau \circ \tau=0,  \label{hr} 
\ee
where the Hirota derivative is defined as \cite{hi} 
\[D_t^m D_x^n f(t,x) \circ g(t,x) =\frac{\pa^m} {\pa a^m} \frac{\pa^n} {\pa b^n} f(t+a,x+b) g(t-a,x-b) |_{a=0, b=0} \] 
and  $ \phi=2 (\ln \tau)_{xx}$. Now, one defines the TNNP for a non-singular web solitons.  
\begin{definition}\cite{jh2, yt}
Let $A$ be a $2n \times 2n$ skew-symmetric matrix. Then $A$ is a totally positive (non-negative) Pfaffian (TNNP) if every $2m \times 2m ( m \leq n ) $ principal sub skew-symmetric matrix of $A$ has a positive (non-negative)  Pfaffian. 
\end{definition}

Given a $2n \times 2n$ TNNP $A=[a_{ij}]$ ,   the $\tau$-function  \cite{ni, or}  of soliton solution of BKP equation (\ref{bkp}) is defined by  
\bea 
\tau (t) &=& 1+ \sum_{ i<j}   a_{ij} \frac{1}{2} \frac{p_i-p_j}{p_i+p_j} e^{\xi (p_i, t)+\xi (p_j , t)}   \no \\
&+ & \sum_{ i<j<k<l} Pf(i,j,k,l) \frac{1}{2^2}  \frac{p_j-p_i}{p_j+p_i} \frac{p_k-p_i}{p_k+p_i} \frac{p_l-p_i}{p_l+p_i} \frac{p_k-p_j}{p_k+p_j}\frac{p_l-p_j}{p_l+p_j}\frac{p_l-p_k}{p_l+p_k} \no \\
&& e^{\xi (p_i, t)+\xi (p_j , t)+ \xi (p_k, t)+ \xi (p_l, t)} \no \\
&+&\sum_{ i<j<k<l<m<q }  Pf(i, j, k, l, m, q)\frac{1}{2^3}\prod_{  r < s \leq q} \frac{p_{r}-p_{s}}{p_{r}+ p_{s}} e^{\xi (p_i, t)+\xi (p_j , t)+ \xi (p_k, t)+ \xi (p_l, t)+\xi (p_m, t)+\xi (p_q, t)} \no \\
&+&  \cdots + Pf (A)   \frac{1}{2^n}\prod_{  r < s \leq 2n} \frac{p_{r}-p_{s}}{p_{r}+ p_{s}} e^{ \sum_{i=1}^{2n}  \xi (p_i, t) },   \label{tc}
\eea
where $\xi (p , t)=e^{p x+p^3 y+ p^5 t}$,  $ p_1^2 > p_2^2> p_3^2 > \cdots >p_{2n}^2,$  and $Pf(i, j, k, l)$ denotes the pfaffian minor of the  skew-symmetric matrix $A=[a_{ij}]$ containing the $ith, jth, kth $  and $lth$ rows, and similar to $Pf(i, j, k, l, m, q)$. The sum in (\ref{tc}) is over all the sub-pfaffians of  $A=[a_{ij}]$ and one defines $Pf(\emptyset )=1$.  Here the pfaffian $Pf(A)$ is defined by 
\bea  Pf(A) &=&  \sum_{\sigma}\epsilon(\sigma)a_{\sigma_1 \sigma_2}a_{\sigma_3
\sigma_4} \cdots a_{\sigma_{2n-1} \sigma_{2n}}.  \label{pf1} 
\eea
Here   $\epsilon(\sigma)$ is the sign of the permutation $\sigma$,  $    \sigma_1 < \sigma_3<  \sigma_5< \cdots < \sigma_{2n-1}$   and $ \sigma_1 <  \sigma_2  , \sigma_3 <  \sigma_4,   \cdots ,   \sigma_{2n-1} <  \sigma_{2 n} $ . \\

  It's known that resonant interaction plays a fundamental  role in multi-dimensional wave phenomenon. The resonances of line solitons of KP-(II) equation
\be    \pa_x (-4 u_t+u_{xxx}+6uu_x)+ 3u_{yy}=0  \label{kp} \ee
has attracted much attractions using the totally non-negative Grassmannians \cite{bc, ko1, ko3}. The resonance of  line solitions is  described by points of the real Grassmannian whose Plucker coordinates are all non-negative.  For the KP-(II) equation (\ref{kp}), the  $\tau$-function  is described by the Wronskian form with respect to $x$-derivatives \cite{hi}.  Inspired by the success of the totally non-negative Grassmannians applied on the resonances of the KP equation \cite{ko6},  it is natural to  consider TNNP applied on the resonances of the BKP equation (\ref{bkp}) .

\indent The paper is organized as follows. In section 2, one constructs TNNP  using the perfect matching,  chord diagram and the Dyck  paths enumeration.  In section 3, we investigate the tridiagonal form.  In section 4, we conclude the paper with several remarks.

\section{ TNNP from Graph Theory}
In this section, we construct TNNP from perfect matching  in graph theory and then a TNNP could also be obtained from any chord diagram. 
\subsection{ Perfect Matching}
\begin{definition}
A planar graph is  drawn on the plane in such a way that its edges intersect only at their endpoints. Given a graph G with edges E and vertices V, a perfect matching in G is a subset M of E, such that every vertex  in V is adjacent to exactly  one edge in M.  The number of perfect matchings is M(G).  If U is a subset of vertices in V(G), then $G\backslash U$ is the subgraph of G induced by the vertices of  $V\backslash U$. 
\end{definition}
\begin{figure}
	\centering
		\includegraphics[width=1.1\textwidth]{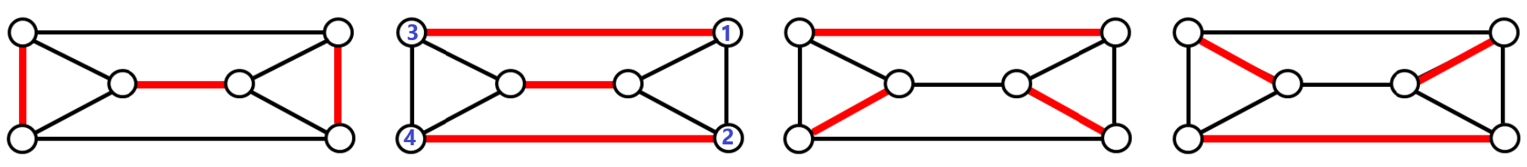}
	\caption{ This plane graph $G$ has six vertices, nine edges and $M(G)=4$. There are four boundary vertices. }
	\end{figure}
For example, please see the figure 1.  \\
\indent Interestingly, the number of perfect matchings of a planar graph with boundary vertices can be found out by the Pfaffian related to the planar graph. 
\begin{theorem}  \cite{ ci, kuo}\\
 Let G be a planar graph with the vertices $a_1, a_2 , ...,a_{2n}$ (boundary vertices) appearing in that cyclic
order among the vertices of some face of G. Consider the skew-symmetric matrix  $ A =
(a_{ij})_{1  \leq i,j \leq 2n} $ with nonzero entries given by
\[  a_{ij} = \left \{ \ba{ll}     M (G\backslash \{ a_i, a_j\} ) ,   \mbox{ if $ i<j $ }   \\
 - M (G\backslash \{ a_i, a_j\} )   ,  \mbox { if $ i>j $ }   \ea       \right.  \]
Then we have that
\[  [M(G)]^{n-1} M(G \backslash \{a_1, a_2  ...,a_{2n}\})=Pf(A). \]
\end{theorem}
\begin{figure}
	\centering
		\includegraphics[width=1\textwidth]{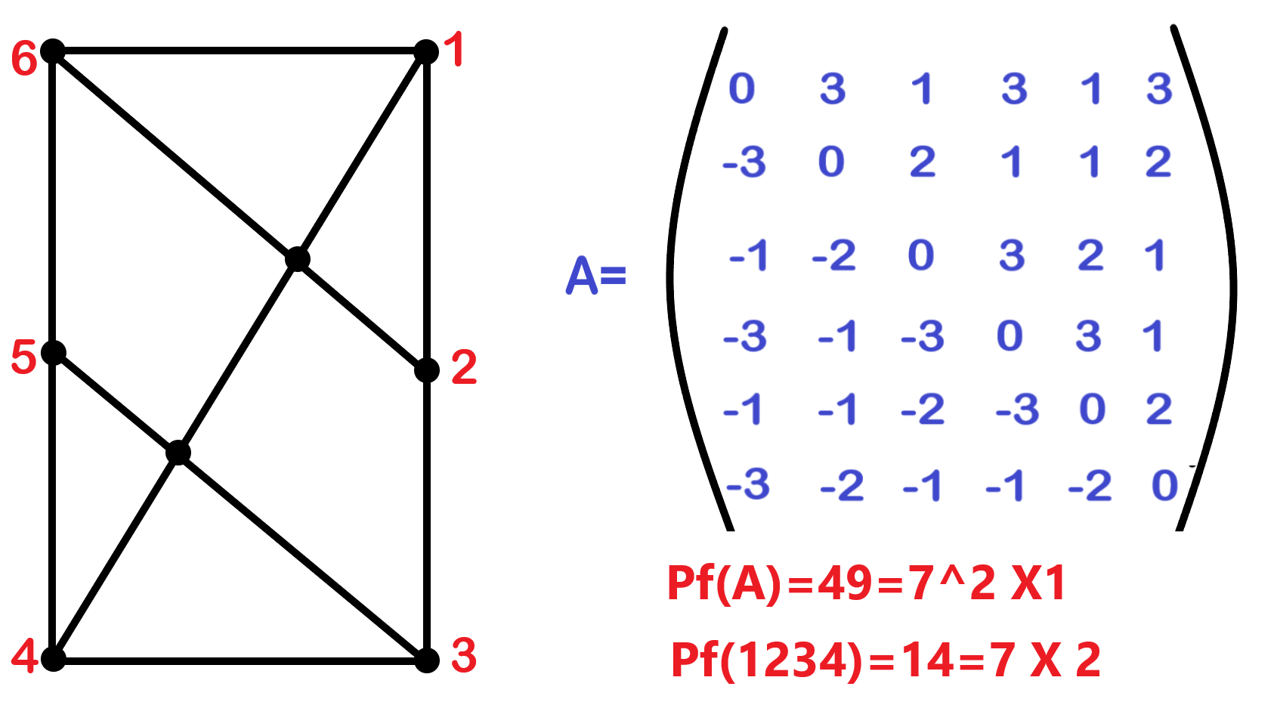}
\caption{ The plane graph $ \textbf {G}$ at left-hand side has eight vertices with six boundary vertices $ (n=3, a_1=1, a_2=2, a_3=3,  \cdots, a_6=6)$, thirteen edges and $M( \textbf {G})=7$. Also, $M(  \textbf{G} \backslash {1,2,3,4})=2$  (k=2) and the matrix $A$ is a TNNP constructed from the graph  $ {\textbf G}$. }	
\end{figure}
From this theorem, we have the following 
\begin{theorem}
For any index $I = \{a_{i_1}, a_{i_2}, \cdots , a_{i_{2k}}\} \subseteq \{a_1, a_2  ...,a_{2n} \} $ in the graph G,  we obtain 
\[   [M(G)]^{k-1} M (G \backslash \{ a_{i_1}, a_{i_2}, \cdots , a_{i_{2k}} \})=Pf(A_I). \]
\end{theorem}
For the figure 1, the  TNNP is 
\[  \left  [\ba{llll} 0 & 1 & 1 &2  \\  - 1 & 0 &2 &1 \\ - 1 & -2 &0 &1 \\  - 2 & -1&-1 &0 \ea  \right ]      \] and its pfaffian is 4. Its soliton graph refers to the figure 1 in \cite{jh2}. The other example is the figure 2. 

Now,  one has 
\begin{theorem}
Given any planar graph $G$ with boundary vertices, we can construct a  totally non-negative pfaffian. Then a soliton solution of the BKP equation (\ref{bkp}) is obtained. 
\end{theorem}
Please see the figure 3. 
\begin{figure}[h]
	\centering
		\includegraphics[width=1.1\textwidth]{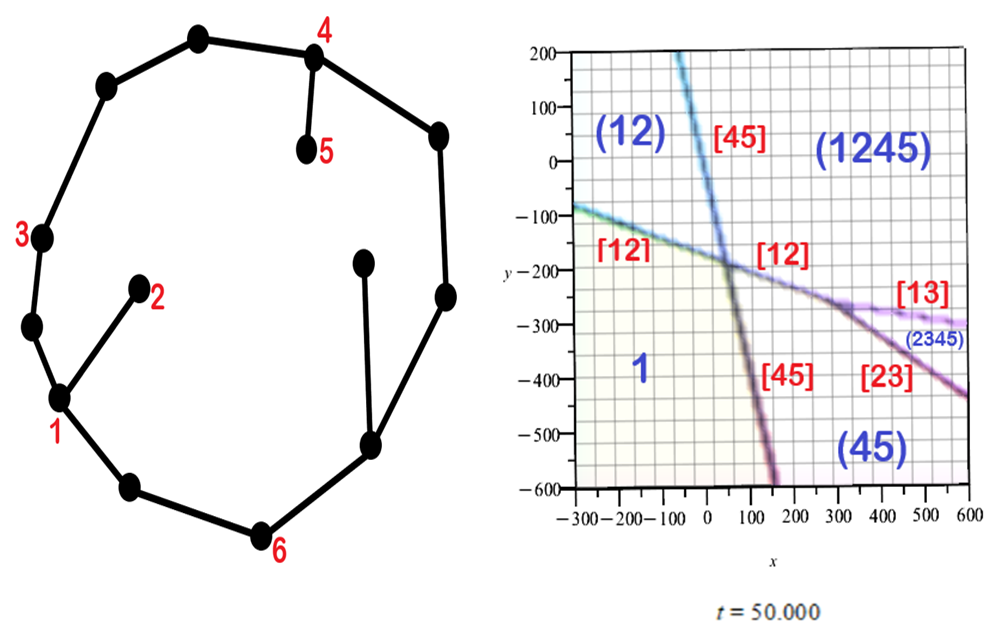}
	\caption{ This plane graph at left-hand side has 14 vertices, 14 edges and $M(G)=1$. There are six numbered boundary vertices. The right-hand is the soliton solution (\ref{bp6})  of BKP from the plane graph with $p_1=2, p_2=1.5, p_3=1, p_4=0.6, p_5=0.3, p_6=0.1$ .  The numbers in blue color are  the sub-pfaffians with dominant exponentials  in the $\tau$-function (\ref{bp6}) at different regions. The numbers in red color
are the $[ij]$-solitons  localized at the boundaries of two distinct regions where a balance exists between two dominant exponentials.	}
	\end{figure}
The singular skew-symmetric matrix associated with the graph is 
\[  \left  [\ba{llllll} 0 & 1 & 0 &0 &0&0   \\ -1 & 0 & 1 &0 &1&1 \\ 0 & -1 & 0 &0 &0&0 \\ 0 & 0& 0 &0 &1&0 
\\ 0 & -1& 0 &-1&0&0 \\0 & -1 & 0 &0 &0&0 \ea  \right ]      \]    
 and its associated $\tau$ function of the BKP equation is 
\bea   
\tau &=&  1+ \frac{1}{2}( \frac{p_1-p_2}{p_1+p_2}e^{\xi (p_1, t)+\xi (p_2 , t)} +    \frac{p_2-p_3}{p_2+p_3}e^{\xi (p_2, t)+\xi (p_3 , t)} +  \frac{p_2-p_5}{p_2+p_5}e^{\xi (p_2, t)+\xi (p_5 , t)}  \no \\
&+ & \frac{p_2-p_6}{p_2+p_6}e^{\xi (p_2, t)+\xi (p_6 , t)} +    \frac{p_4-p_5}{p_4+p_5}e^{\xi (p_4, t)+\xi (p_5 , t)}   ) \no \\
&+&\frac{1}{4} \frac{p_1-p_2}{p_1+p_2}\frac{p_4-p_5}{p_4+p_5}e^{\xi (p_1, t)+\xi (p_2 , t)+\xi (p_4, t)+\xi (p_5 , t)} +\frac{1}{4} \frac{p_2-p_3}{p_2+p_3}\frac{p_4-p_5}{p_4+p_5}e^{\xi (p_2, t)+\xi (p_3, t)+\xi (p_4, t)+\xi (p_5 , t)} \no \\
&+&\frac{1}{4} \frac{p_2-p_6}{p_2+p_6}\frac{p_4-p_5}{p_4+p_5}e^{\xi (p_2, t)+\xi (p_3, t)+\xi (p_4, t)+\xi (p_5 , t)}       \label{bp6}. 
\eea

As an application of this theorem, one considers the chord diagram. 
\begin{definition}
A chord diagram is a diagram of N chords joining $2N$ points on a circle  in disjoint pairs. 
\end{definition}
\begin{theorem}
Given a chord diagram, one can construct a TNNP. 
\end{theorem}
\begin{proof}
The proof is simple. One considers the following figures:
\begin{figure}[h]
	\centering
\includegraphics[width=0.5\textwidth]{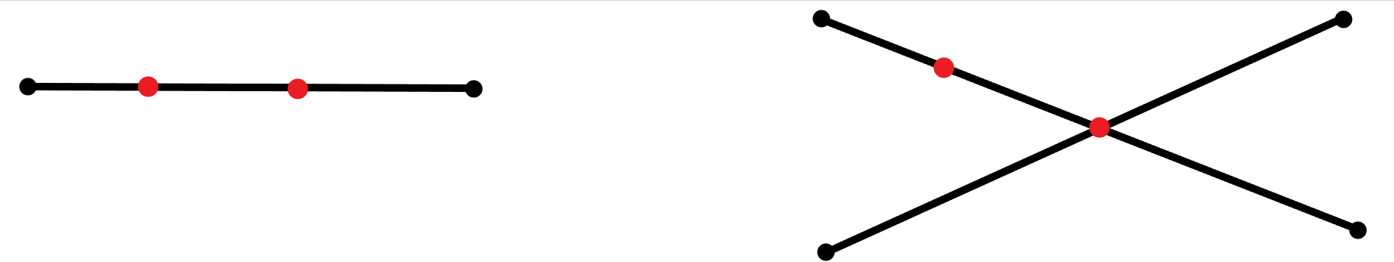}
	\end{figure}
If there is no intersection, we add two points to the lines. If there are  even intersection points, one will add one point to the each  intersection. If there are odd intersection points, one will add one extra point to a line such that the  interior points  are even number. Then the Kuo-Ciucu Theorem is applied. 
\end{proof}
More examples can be seen in the figure 4. \\
\begin{figure}
	\centering
\includegraphics[width=1\textwidth]{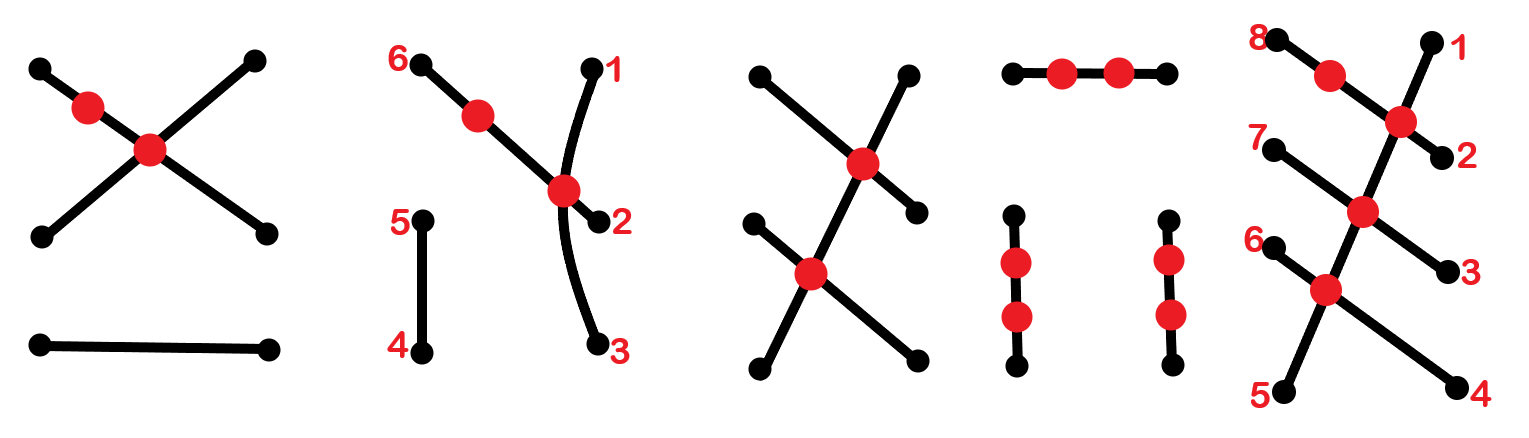}
\caption{The  interior points  are even number.} 
	\end{figure}
\textbf{Remarks:}\\
(1) The number of all chord diagrams of N chords is $ 1 \times 3 \times 5 \times 7 \cdots \times (2N-1) =\frac{(2N) !}{2^N (N!)}$. \\
(2) There are three types of chord diagrams: P-type (nesting), O-type (alignment) and T-type (crossing).They are used to describe elastic resonant solitons of water waves in the KP equation \cite{ck}. The number of P-type and  O-type chord diagrams (non-crossing) of N chords  is the Catalan number $C(N)=\frac{1}{N+1} {2N \choose N }$. \\
(3)Given a chord diagram, constructing a TNNP is non-unique. \\

\subsection{Lattice Path Enumeration }
A lattice path in 2-dimensional Euclidean space is a path in a lattice plane $\mathbb {Z}^2 $. Given two lattice points $ (2a, 0)$ and $ (2b, 0)$ , $a,b \in \mathbb{Z}$ and $ a \leq b $  , one considers the steps  (1, 1)  or  (1, -1)  that never pass below the x-axis in the lattice plane.  They are called the Dyck paths \cite{cr}. It is shown that the number of Dyck paths is the Catalan number $C(b-a)=\frac{1}{b-a+1} {2(b-a) \choose (b-a) }$, where one defines $C(0)=0$. Please see the figure 5. 
\begin{figure}[h]
	\centering
		\includegraphics[width=0.9\textwidth]{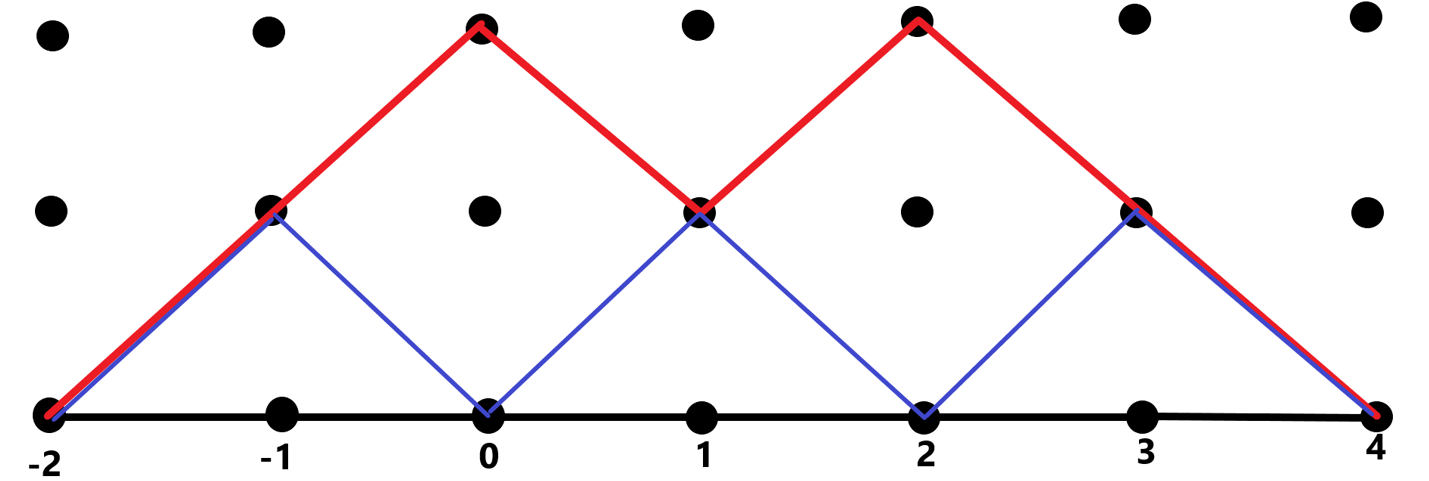}
\caption{Two Dyck paths are shown here with a=-1, b=2. There are $C(3)=5$ Dyck paths.} 	
\end{figure}

We have the following 
\begin{theorem}(\cite{st}, p.265)\\
Let $ a_1 \leq a_2 \leq \cdots \leq a_{2n} $ be integers and $C( a_1, a_2, a_3, \cdots, a_{2n}) $ denote the number of 
non-self intersecting Dyck paths $ ( L_1, L_2, L_3, \cdots, L_n) $ to connect the points 
$ (2a_1, 0), ( 2a_2, 0),  (2a_3, 0),  \cdots, (2a_{2n},0)$.  Then 
\be C( a_1, a_2, a_3, \cdots, a_{2n})  = Pf (C(a_j-a_i)).   \label{ca} \ee
\end{theorem}
So one has 
\begin{theorem}
The skew symmetric $ 2n \times 2n$ matrix $C(a_j-a_i)$ in (\ref {ca}) is a TNNP. 
\end{theorem}
\begin{proof}
Let $I=\{k_1, k_2, \cdots, k_{2m} \}  \subset  \{ a_1,  a_2 , \cdots , a_{2n} \} $  and $ k_1 \leq k_2   \leq  \cdots  \leq k_{2m} $.  Then $Pf (C(k_j-k_i))$ is the number 
of  non-self intersecting Dyck paths $ (L_1, L_2, L_3, \cdots, L_m ) $ to connect the points 
$ (2k_1, 0), ( 2k_2, 0),  (2k_3, 0),  \cdots, (2k_{2m},0)$. 
\end{proof} 
For example, 
\bean  C( -2, 1, 3, 4) &=& Pf  (\left  [\ba{llll} 0 & C(3) & C(5)  &C(6)  \\  -  C(3) & 0 &C(2)  &C(3) \\  - C(5)& -C(2 ) &0 & C(1)  \\  -  C(6) & -C(3) & -  C(1) &0 \ea  \right ]   )   \\
&=&  Pf  (\left  [\ba{llll} 0 & 5 & 42  & 132  \\  -  5 & 0 &2 & 5 \\  - 42 & -2  &0 & 1  \\  -  132  & -5 & -  1 &0 \ea  \right ]   )=59. 
\eean
So there are 59 sets of non-self intersecting Dyck paths $ (L_1^{(\alpha)}, L_2^{(\alpha) }, L_3^{ (\alpha)}, L_4^{ (\alpha) }), \alpha=1,2,3, \cdots, 59, $ to connect the points $ (-4, 0), ( 2, 0),  (6, 0),   (8,0)$. Its soliton graph also refers to the figure 1 in \cite{jh2}.  \\

The merit of the Dyck-path  representation (\ref{ca}) is that one can easily extend a $2n \times 2n $ TNNP to a 
$2(n+1)  \times 2(n+1)  $ TNNP by adding two more points $ \{ a_{2n+1}, a_{2n+2}\}, a_{2n+1} \leq  a_{2n+2}$. For instance, we can extend $ C( -2, 1, 3, 4) $ to the singular TNNP
\bean 
C( -2, 1, 3, 4, 4, 6) &=& Pf  (\left  [\ba{llllll} 0 & C(3) & C(5)  &C(6) &C(6) &C(8) \\  -  C(3) & 0 &C(2)  &C(3) & C(3) &C(5 ) \\  - C(5)& -C(2 ) &0 & C(1) & C(1) & C(3 ) \\  -  C(6) & -C(3) & -  C(1) &0 &C(0)  &C(2) \\ 
-  C(6) & -C(3) & -  C(1) &-C(0 ) &0  &C(2) \\ -  C(8) & -C(5) & -  C(3) & -  C(2) &-C(2)  &0 \ea  \right ]   )   \\
&=&  Pf  (\left    [\ba{llllll} 0 & 5 & 42  &132 &132 &1430 \\  - 5 & 0 &2 &5 & 5 &42 \\  - 42& -2  &0 & 1 & 1 & 5 \\  -  132 & -5 & -  1 &0 &0  &2 \\ 
- 42 & -5 & -  1 &0 &0  &2\\ - 1430 & -42 & - 5 & - 2 &-2  &0 
\ea  \right ]   )=0. 
\eean 
Furthermore,  one obtains the  sub-pfaffian 
\bean  C( -2, 1, 4, 6) &=& Pf  (\left  [\ba{llll} 0 & C(3) & C(6)  &C(8)  \\  -  C(3) & 0 &C(3)  &C(5) \\  - C(6)& -C(3 ) &0 & C(2)  \\  -  C(8) & -C(5) & -  C(2) &0 \ea  \right ]   )   \\
&=&  Pf  (\left  [\ba{llll} 0 & 5 & 42  & 1430  \\  -  5 & 0 &5 & 42\\  - 42 & -5  &0 & 2  \\  -  1430  & -42 & -  2&0 \ea  \right ]   )=5396. 
\eean
Similarly,  there are 5396 sets of non-self intersecting Dyck paths $ (L_1^{(\alpha)}, L_2^{(\alpha) }, L_3^{ (\alpha)}, L_4^{ (\alpha) }), \alpha=1,2,3, \cdots, 5396, $ to connect the points $ (-4, 0), ( 2, 0),  (8, 0),   (12,0)$. 

In Theorem 2.8, one can consider the strict case. 
\begin{theorem}
Let $ a_1 <  a_2 < \cdots <  a_{2n} $ be integers. Then  $C(a_j-a_i)$  is a totally positive pfaffian.
\end{theorem}
\begin{proof}
Let $I=\{k_1, k_2, \cdots, k_{2m} \}  \subset  \{ a_1,  a_2 , \cdots , a_{2n} \} $ and  $ k_1 <  k_2 < \cdots <  k_{2m}. $Then $Pf (C(k_j-k_i))$ is the number of  non-self intersecting Dyck paths $ (L_1, L_2, L_3, \cdots, L_m ) $ to connect the points $ (2k_1, 0), ( 2k_2, 0),  (2k_3, 0),  \cdots, (2k_{2m},0)$. Since $ k_1 <  k_2 < \cdots <  k_{2m} $, the number is non-zero  due to nestings or alignments.  Then each sub-pfaffian is positive. This completes the proof.
\end{proof} 
\section{ Reduction to the Tridiagonal Form}
In this section, one considers the  reduction of a skew-symmetric matrix to the tridiagonal form using the Gauss transformation \cite{bu, wi}. From this reduction, one could construct a TNNP and its extension. 

Let A be a $ 2n \times 2n $ skew-symmetric matrix and we have the factorization 
\be  A= LDL^T , \label{fa} \ee 
where L is a unit  lower-triangular matrix with first column as $\vec 0$ and D is a skew-symmetric tridiagonal matrix, that is, 
\bean 
L &=& \left  [\ba{lllllll} 1  & 0 & 0  &0 &0  &\cdots  &0  \\  0 & 1 &0   &0 & 0& \cdots  &0 \\  0& l_{32}  &1  & 0  & 0& \cdots & 0  \\  0 & l_{42}   & l_{43} &1 & 0  &\cdots  &0 \\ 
0 & \vdots   &\vdots   & \vdots & \vdots   &\vdots \\ 0 & l_{2n-1, 2}   &  l_{2n-1, 3} &\cdots & l_{2n-1, 2n-2} &1 &0 \\
0 &   l_{2n, 2}   &  l_{2n, 3} & l_{2n, 4}& \cdots  & l_{2n, 2n-1}&1   \ea  \right ]     \\
\eean 
and 
\bean 
 D &=& \left  [\ba{lllllll} 0 & d_{12} & 0  &0  &0  &\cdots &  0  \\  - d_{12}  & 0 & d_{23}   &0 & 0  &\cdots &0  \\  
0 &  - d_{23} & 0 & d_{34} & 0 & \cdots&0\\  0  & 0 & - d_{34}  &0 &d_{45} &\cdots &0 \\ 
  0 & \vdots   &\vdots   & \vdots & \vdots   &\vdots & \vdots \\ 0 & 0 & 0  &0 &-d_{2n-2,2n-1} &   0  & d_{2n-1,2n}\\
0 &0  & 0  & \cdots &0  &- d_{2n-1,2n} &0  \ea  \right ] .      \\
\eean 
Also, one has 
\[ Pf(A)=d_{12} d_{34} d_{56} \cdots d_{2n-1,2n}, \]
and  the expansion (Lemma 2.1 in \cite{is}) 
\bea
a_{ij} & = & \sum_{1 \leq \alpha < \beta \leq 2n} d_{ \alpha  \beta } \quad det \left  [\ba{ll}  l_{i \alpha}   & l_{i \beta}  \\   l_{j \alpha} &  l_{j \beta}  \ea  \right ]  \no \\
&=&  \sum_{k=1 }^n  d_{ 2k-1, 2k } det \left  [\ba{ll}  l_{i , 2k-1}   & l_{i , 2k }  \\   l_{j , 2k-1} &  l_{j , 2k }  \ea  \right ] +   \sum_{k=1 }^{n-1}   d_{ 2k, 2k+1 } det \left  [\ba{ll}  l_{i , 2k}   & l_{i , 2k+1 }  \\   l_{j , 2k} &  l_{j , 2k+1  }  \ea  \right ] \no \\
&=& \sum_{k=1 }^{n-1} \left (  d_{ 2k-1, 2k } det \left  [\ba{ll}  l_{i , 2k-1}   & l_{i , 2k }  \\   l_{j , 2k-1} &  l_{j , 2k }  \ea  \right ] +d_{ 2k, 2k+1 } det \left  [\ba{ll}  l_{i , 2k}   & l_{i , 2k+1 }  \\   l_{j , 2k} &  l_{j , 2k+1  }  \ea  \right ]\right)+ d_{ 2n-1, 2n }. \no \\             \label{ex}
\eea  

For a skew-symmetric matrix A, we denote $\om_I$  as  the sub-pfaffian $Pf(A_I)$, where $ I=\{a_{k_1}, a_{k_2}, \cdots, a_{k_{2m}} \}  \subset  \{ a_1,  a_2 , \cdots , a_{2n} \} $ and  $ k_1 <  k_2 < \cdots <  k_{2m}$. We notice that  $\om_{ij}=a_{ij}$. One has the following pfaffian identity (p.92 in \cite{hi})
\be  \om_I   \om_{I pqrs}=  \om_{Ipq}   \om_{Irs} - \om_{Ipr}   \om_{Iqs}    +   \om_{Ips}   \om_{Iqr}.\label{id}  \ee
The following theorem gives the so-called Pfaﬃan factorization of a skew-symmetric matrix A using the sub-pfaffians.
\begin{theorem}
In the factorization (\ref{fa}), $n \geq 2$, one has,   
\bea
l_{i,2k} &=& \frac{ \om_{\widehat{1, 2k-1},i}}  { \om_{\widehat{1, 2k}}},  \quad   l_{i, 2k-1}=
\frac{ \om_{\widehat{2, 2k-2},i}}  { \om_{\widehat{2, 2k-1}}} ,  \label{cp} \\
d_{2k-1, 2k} &=& \frac{ \om_{\widehat{1, 2k}}}   { \om_{\widehat{1, 2k-2}}},  \quad   d_{2k, 2k+1}=
\frac{ \om_{\widehat{2, 2k+1}}}  { \om_{\widehat{2, 2k-1}}},  \label{co}
\eea
where  $k =1,  2, 3, 4 , \cdots, n $ and $\widehat{a , b}= \{a, a+1, a+2, \cdots , b \}$. Here one defines 
\[ \om_{\widehat{2, 0}} =0 ,  \quad  \om_{\widehat{1, 0}}=\om_{\widehat{2, 1}}=1. \] 
\end{theorem}
For instance, n=6, we have 
\bean 
L &=& \left  [\ba{llllll} 1  & 0 & 0  &0 &0  &0  \\  0 & 1 &0   &0 & 0  &0 \\  0 & \frac{ \om_{13} }  { \om_{12}}   &1 &  0 &  0 & 0  \\  0 &  \frac{ \om_{1 4} }  { \om_{12}} & \frac{ \om_{24} }  { \om_{23}}    &1 & 0  &0 \\ 
0 & \frac{ \om_{15} }  { \om_{12}}    & \frac{ \om_{25} }  { \om_{23}} &   \frac{ \om_{1235} }  { \om_{1234}}              & 1  &0 \\    0 &   \frac{ \om_{16} }  { \om_{12}}  &  \frac{ \om_{26} }  { \om_{23}} &\frac{ \om_{1236} }  { \om_{1234}}  &  \frac{ \om_{2346} }  { \om_{2345}} &1   \ea  \right ]  =  \left  [\ba{llllll} 1  & 0 & 0  &0 &0  &0  \\  0 & 1 &0   &0 & 0  &0 \\  0 & \frac{ \om_{13} }  { \om_{12}}   &1 &  0 &  0 & 0  \\  0 &  \frac{ \om_{1 4} }  { \om_{12}} & \frac{ \om_{24} }  { \om_{23}}    &1 & 0  &0 \\ 
0 & \frac{ \om_{15} }  { \om_{12}}    & \frac{ \om_{25} }  { \om_{23}} &    \frac{  \om_{\widehat{1, 3},5}} {
 \om_{\widehat{1, 4}} } & 1  &0 \\    0 &   \frac{ \om_{16} }  { \om_{12}}  &  \frac{ \om_{26} }  { \om_{23}} &\frac{  \om_{\widehat{1, 3},6}}  {  \om_{\widehat{1, 4}}}  &  \frac{  \om_{\widehat{2, 4},6} }  {  \om_{\widehat{2, 5}}} &1   \ea  \right ]  
\eean 
and 
\bean 
 D &=& \left  [\ba{llllll} 0 & \om_{12}  & 0  &0  &0 &  0  \\  - \om_{12} & 0 & \om_{23}   &0 & 0 &0  \\  0 &  - \om_{23} & 0 & \frac{ \om_{1234} }  { \om_{12}} & 0 & 0 \\  0  & 0 & -  \frac{ \om_{1234} }  { \om_{12}}&0 &   \frac{ \om_{2345} }  { \om_{23}} &0 \\  0 & 0  &0 &-  \frac{ \om_{2345} }  { \om_{23}}     &   0  & 
 \frac{ \om_{123456} }  { \om_{1234}}  \\ 0 &0  & 0  &0  &-  \frac{ \om_{123456} }  { \om_{1234}} &0  \ea  \right ] \\
&=&   \left  [\ba{llllll} 0 & \om_{12}  & 0  &0  &0 &  0  \\  - \om_{12} & 0 & \om_{23}   &0 & 0 &0  \\  0 &  - \om_{23} & 0 & \frac{  \om_{\widehat{1, 4}}  }  { \om_{12}} & 0 & 0 \\  0  & 0 & -  \frac{  \om_{\widehat{1, 4}}}  { \om_{12}}&0 &   \frac{  \om_{\widehat{2, 5}} }  { \om_{23}} &0 \\  0 & 0  &0 &-  \frac{  \om_{\widehat{2, 5}} }  { \om_{23}}     &   0  & 
 \frac{\om_{\widehat{1,6}  }}  {  \om_{\widehat{1, 4}}}  \\ 0 &0  & 0  &0  &-  \frac{  \om_{\widehat{1,6}}  }{  \om_{\widehat{1, 4}}} &0  \ea  \right ]    .      
\eean 
\begin{proof}
We use induction on n. When $n=2$, it is a direct calculation. Assume our claim holds when the matrix size of $A$ is $n-1$. Then by (\ref{ex}), (\ref{cp}) and (\ref{co}) , one has  
\bea
a_{2i, 2i+1} & = & \sum_{k=1 }^{i} \frac{\om_{\widehat{1, 2k}}}{\om_{\widehat{1, 2k-2}}}    \quad det \left  [\ba{ll}    \frac{\om_{\widehat{2, 2k-2}, 2i}}{\om_{\widehat{2, 2k-1}}}  & \frac{\om_{\widehat{1, 2k-1},2i}}{\om_{\widehat{1, 2k}}}  \\    \frac{\om_{\widehat{2, 2k-2}, 2i+1}}{\om_{\widehat{2, 2k-1}}}                                                &                                 \frac{\om_{\widehat{1, 2k-1}, 2i+1}}{\om_{\widehat{1, 2k}}}  \ea  \right ]  \no \\
&+& \sum_{k=1 }^{i-1} \frac{\om_{\widehat{2, 2k+1}}}{\om_{\widehat{2, 2k-1}}}    \quad det \left  [\ba{ll}   
 \frac{\om_{\widehat{1, 2k-1}, 2i}}{\om_{\widehat{1, 2k}}}  & \frac{\om_{\widehat{2, 2k},2i}}{\om_{\widehat{2, 2k+1}}}  \\    \frac{\om_{\widehat{1, 2k-1}, 2i+1}}{\om_{\widehat{1, 2k}}}                                                &                                 \frac{\om_{\widehat{2, 2k}, 2i+1}}{\om_{\widehat{2, 2k+1}}}  \ea  \right ]   
+ \frac{\om_{\widehat{2, 2i+1}}}{\om_{\widehat{2, 2i-1}}}.   \label{id}
\eea  

Replacing the 2i-th  row/column by the r-th  row/column and the (2i+1)-st  row/column by the s-th row/column in this identity (\ref{id}), we see
\bea
a_{r, s} & = & \sum_{k=1 }^{i} \frac{\om_{\widehat{1, 2k}}}{\om_{\widehat{1, 2k-2}}}    \quad det \left  [\ba{ll}    \frac{\om_{\widehat{2, 2k-2}, r}}{\om_{\widehat{2, 2k-1}}}  & \frac{\om_{\widehat{1, 2k-1},r}}{\om_{\widehat{1, 2k}}}  \\    \frac{\om_{\widehat{2, 2k-2}, s}}{\om_{\widehat{2, 2k-1}}}                                                &                                 \frac{\om_{\widehat{1, 2k-1}, s}}{\om_{\widehat{1, 2k}}}  \ea  \right ]  \no \\
&+& \sum_{k=1 }^{i-1} \frac{\om_{\widehat{2, 2k+1}}}{\om_{\widehat{2, 2k-1}}}    \quad det \left  [\ba{ll}   
 \frac{\om_{\widehat{1, 2k-1}, r}}{\om_{\widehat{1, 2k}}}  & \frac{\om_{\widehat{2, 2k},r}}{\om_{\widehat{2, 2k+1}}}  \\    \frac{\om_{\widehat{1, 2k-1}, s}}{\om_{\widehat{1, 2k}}}                                                &                                 \frac{\om_{\widehat{2, 2k}, s}}{\om_{\widehat{2, 2k+1}}}  \ea  \right ]   
+ \frac{\om_{\widehat{2, 2i-1},r,s }}{\om_{\widehat{2, 2i-1}}}   \label{rs}
\eea
for  any $ r, s$. 

On the other hand, from (\ref{ex}), it can be seen that $l_{s, r}$ , $ 2 \leq i <n, s=2n-1, 2n, r=2i$, must satisfy the linear equations
\bea 
a_{2i, s} &=&  \sum_{k=1 }^{i} \frac{\om_{\widehat{1, 2k}}}{\om_{\widehat{1, 2k-2}}}    \quad det \left  [\ba{ll}    \frac{\om_{\widehat{2, 2k-2}, 2i}}{\om_{\widehat{2, 2k-1}}}  & \frac{\om_{\widehat{1, 2k-1},2i}}{\om_{\widehat{1, 2k}}}  \\    \frac{\om_{\widehat{2, 2k-2}, s}}{\om_{\widehat{2, 2k-1}}} & \frac{\om_{\widehat{1, 2k-1}, s}}{\om_{\widehat{1, 2k}}}  \ea  \right ]\\
&+& \sum_{k=1 }^{i-1} \frac{\om_{\widehat{2, 2k+1}}}{\om_{\widehat{2, 2k-1}}}    \quad det \left  [\ba{ll}   
 \frac{\om_{\widehat{1, 2k-1}, 2i}}{\om_{\widehat{1, 2k}}}  & \frac{\om_{\widehat{2, 2k},2i}}{\om_{\widehat{2, 2k+1}}}  \\    \frac{\om_{\widehat{1, 2k-1}, s}}{\om_{\widehat{1, 2k}}} & \frac{\om_{\widehat{2, 2k}, s}}{\om_{\widehat{2, 2k+1}}}  \ea  \right ] +\frac{\om_{\widehat{2, 2i+1}}}{\om_{\widehat{2, 2i-1}}}l_{s, 2i+1}    \no 
\eea
Comparing this equation with (\ref{rs}), it's seen that $l_{s, 2i+1}=\frac{ \om_{\widehat{2, 2i},s}}  { \om_{\widehat{2, 2i+1}}}$.  

As for $r=2i+1$, $ 2 \leq i <n, s=2n-1, 2n$, one has from (\ref{ex})
\bea
a_{2i+1, s} & = & \sum_{k=1 }^{i} \frac{\om_{\widehat{1, 2k}}}{\om_{\widehat{1, 2k-2}}}    \quad det \left  [\ba{ll}    \frac{\om_{\widehat{2, 2k-2}, 2i+1}}{\om_{\widehat{2, 2k-1}}}  & \frac{\om_{\widehat{1, 2k-1},2i+1}}{\om_{\widehat{1, 2k}}}  \\    \frac{\om_{\widehat{2, 2k-2}, s}}{\om_{\widehat{2, 2k-1}}}                                                &                                 \frac{\om_{\widehat{1, 2k-1}, s}}{\om_{\widehat{1, 2k}}}  \ea  \right ] 
 + \frac{\om_{\widehat{1, 2i+2}}}{\om_{\widehat{1, 2i}}}l_{s, 2i+2}\no \\
&+& \sum_{k=1 }^{i-1} \frac{\om_{\widehat{2, 2k+1}}}{\om_{\widehat{2, 2k-1}}}    \quad det \left  [\ba{ll}   
 \frac{\om_{\widehat{1, 2k-1}, 2i+1}}{\om_{\widehat{1, 2k}}}  & \frac{\om_{\widehat{2, 2k},2i+1}}{\om_{\widehat{2, 2k+1}}}  \\    \frac{\om_{\widehat{1, 2k-1}, s}}{\om_{\widehat{1, 2k}}}                                                &                                 \frac{\om_{\widehat{2, 2k}, s}}{\om_{\widehat{2, 2k+1}}}  \ea  \right ] + \frac{\om_{\widehat{2, 2i+1}}}{\om_{\widehat{2, 2i-1}}} \quad det \left  [\ba{ll}   
 \frac{\om_{\widehat{1, 2i-1}, 2i+1}}{\om_{\widehat{1, 2i}}}  & 1 \\    \frac{\om_{\widehat{1, 2i-1}, s}}{\om_{\widehat{1, 2i}}}   &   \frac{\om_{\widehat{2, 2i}, s}}{\om_{\widehat{2, 2i+1}}}  \ea  \right ] \no  \\  \label{ec}
\eea
Comparing this equation with (\ref{rs}), we obtain that 
\be
\frac{\om_{\widehat{1, 2i+2}}}{\om_{\widehat{1, 2i}}}l_{s, 2i+2}+  \frac{\om_{\widehat{2, 2i+1}}}{\om_{\widehat{2, 2i-1}}}det \left  [\ba{ll}   \frac{\om_{\widehat{1, 2i-1}, 2i+1}}{\om_{\widehat{1, 2i}}}  & 1 \\    \frac{\om_{\widehat{1, 2i-1}, s}}{\om_{\widehat{1, 2i}}}   &   \frac{\om_{\widehat{2, 2i}, s}}{\om_{\widehat{2, 2i+1}}}  \ea  \right ]=
\frac{\om_{\widehat{2, 2i-1},2i+1,s }}{\om_{\widehat{2, 2i-1}}}.  \label{so} 
\ee
Using the identity (\ref{id}), one yields
\bea 
&&\om_{\widehat{2, 2i-1}}\om_{\widehat{1,2i+1}, s}= \om_{\widehat{2, 2i-1}}\om_{\widehat{2, 2i-1}, 2i, 2i+1, s, 1} \no  \\ 
&=&\om_{\widehat{2, 2i-1}, 2i, 2i+1}\om_{\widehat{2, 2i-1}, s, 1} -\om_{\widehat{2, 2i-1}, 2i, s}\om_{\widehat{2, 2i-1}, 2i+1, 1}+  \om_{\widehat{2, 2i-1}, 2i, 1}\om_{\widehat{2, 2i-1}, 2i+1, s} \no   \\
&=& \om_{\widehat{2, 2i+1}}\om_{\widehat{1, 2i-1}, s} -\om_{\widehat{2, 2i}, s}\om_{\widehat{1, 2i-1}, 2i+1}+  \om_{\widehat{1, 2i}}\om_{\widehat{2, 2i-1}, 2i+1, s}.  \label{faa}
\eea 
Plugging it into (\ref{so}), we have $l_{s, 2i+2}=\frac{ \om_{\widehat{1, 2i+1},s}}  { \om_{\widehat{1, 2i+2}}}$ after a simple calculation.

Finally,  letting $i=n-1,  s=2n$  in  (\ref{ec}),  one  has 
\bea
a_{2n-1, 2n} & = & \sum_{k=1 }^{n-1} \frac{\om_{\widehat{1, 2k}}}{\om_{\widehat{1, 2k-2}}}    \quad det \left  [\ba{ll}    \frac{\om_{\widehat{2, 2k-2}, 2n-1}}{\om_{\widehat{2, 2k-1}}}  & \frac{\om_{\widehat{1, 2k-1},2n-1}}{\om_{\widehat{1, 2k}}}  \\    \frac{\om_{\widehat{2, 2k-2}, 2n}}{\om_{\widehat{2, 2k-1}}}                                                &                                 \frac{\om_{\widehat{1, 2k-1}, 2n}}{\om_{\widehat{1, 2k}}}  \ea  \right ] 
 + d_{2n, 2n-1}  \no \\
&+& \sum_{k=1 }^{n-2} \frac{\om_{\widehat{2, 2k+1}}}{\om_{\widehat{2, 2k-1}}}    \quad det \left  [\ba{ll}   
 \frac{\om_{\widehat{1, 2k-1}, 2n-1}}{\om_{\widehat{1, 2k}}}  & \frac{\om_{\widehat{2, 2k},2n-1}}{\om_{\widehat{2, 2k+1}}}  \\    \frac{\om_{\widehat{1, 2k-1}, 2n}}{\om_{\widehat{1, 2k}}}                                                &                                 \frac{\om_{\widehat{2, 2k}, 2n}}{\om_{\widehat{2, 2k+1}}}  \ea  \right ] + \frac{\om_{\widehat{2, 2n-1}}}{\om_{\widehat{2, 2n-3}}} \quad det \left  [\ba{ll}   
 \frac{\om_{\widehat{1, 2n-3}, 2n-1}}{\om_{\widehat{1, 2n-2}}}  & 1 \\    \frac{\om_{\widehat{1, 2n-3}, 2n}}{\om_{\widehat{1, 2n-2}}}   &   \frac{\om_{\widehat{2, 2n-2}, 2n}}{\om_{\widehat{2, 2n-1}}}  \ea  \right ] \no  \\  \label{kc}
\eea
On the other hand, letting $r=2n-1, s=2n$ in  (\ref{rs}),  one also has 
\bea
a_{2n-1, 2n} & = & \sum_{k=1 }^{n-1} \frac{\om_{\widehat{1, 2k}}}{\om_{\widehat{1, 2k-2}}}    \quad det \left  [\ba{ll}    \frac{\om_{\widehat{2, 2k-2}, 2n-1}}{\om_{\widehat{2, 2k-1}}}  & \frac{\om_{\widehat{1, 2k-1},2n-1}}{\om_{\widehat{1, 2k}}}  \\    \frac{\om_{\widehat{2, 2k-2}, 2n}}{\om_{\widehat{2, 2k-1}}}                                                &                                 \frac{\om_{\widehat{1, 2k-1}, 2n}}{\om_{\widehat{1, 2k}}}  \ea  \right ] 
   \no \\
&+& \sum_{k=1 }^{n-2} \frac{\om_{\widehat{2, 2k+1}}}{\om_{\widehat{2, 2k-1}}}    \quad det \left  [\ba{ll}   
 \frac{\om_{\widehat{1, 2k-1}, 2n-1}}{\om_{\widehat{1, 2k}}}  & \frac{\om_{\widehat{2, 2k},2n-1}}{\om_{\widehat{2, 2k+1}}}  \\    \frac{\om_{\widehat{1, 2k-1}, 2n}}{\om_{\widehat{1, 2k}}}  &   \frac{\om_{\widehat{2, 2k}, 2n}}{\om_{\widehat{2, 2k+1}}}  \ea  \right ] + \frac{\om_{\widehat{2, 2n-3}, 2n-1, 2n }}{\om_{\widehat{2, 2n-3}}}   \label{fc}
\eea
Comparing (\ref{kc}) with (\ref{fc}), we obtain 
\[  d_{2n, 2n-1}+  \frac{\om_{\widehat{2, 2n-1}}}{\om_{\widehat{2, 2n-3}}} \quad det \left  [\ba{ll}   
 \frac{\om_{\widehat{1, 2n-3}, 2n-1}}{\om_{\widehat{1, 2n-2}}}  & 1 \\    \frac{\om_{\widehat{1, 2n-3}, 2n}}{\om_{\widehat{1, 2n-2}}}   &   \frac{\om_{\widehat{2, 2n-2}, 2n}}{\om_{\widehat{2, 2n-1}}}  \ea  \right ]  =\frac{\om_{\widehat{2, 2n-3}, 2n-1, 2n }}{\om_{\widehat{2, 2n-3}}} .   \]
A similar argument to (\ref{faa}), one has  $ d_{2n, 2n-1}= \frac{\om_{\widehat{1, 2n}}}{\om_{\widehat{1, 2n-2}}}$. This completes the proof.
\end{proof} 
One remarks that $L$ and $D$ can be  composed of  the $ 2 \times 2$ blocks and investigated  to evaluate a pfaﬃan analogue of q-Catalan Hankel determinants \cite{iz}. 

To construct TNNP from the factorization (\ref{fa}), we can assume each entry in $D$ is non-negative and $L$ is a totally non-negative $(0,1)$-matrix, that is, each entry in $L$ is 0 or 1 and each minor of $L$  is non-negative. To obtain a totally non-negative $(0,1)$-matrix  $L$,  we consider the following double echelon form \cite{fj}:\\
(1) each row has the form $ (0,0,\cdots, 0, 1,1 , \cdots, 1, 0,  \cdots,0)$ except the first row. \\
(2)The first and last nonzero entries in row $i+1$  are not to the left of the first and last nonzero entries in row $i$, respectively (   $i  = 2,3, \cdots  , 2n-1 $ ). It is shown (p.40 in \cite{fj} ) that if $L$ is of double echelon form and doesn't contain  the sub-matrix 
\be  \left  [\ba{lll} 1 & 1 &  0  \\ 1& 1  &1  \\  0 & 1 & 1  \ea \right ],  \label{ta} \ee
then $L$ is a totally non-negative $(0,1)$-matrix. For example, 
 \[  \left  [\ba{llllll} 1  & 0 & 0  &0 &0  &0  \\  0 & 1 &0   &0 & 0  &0 \\  0 & 1   &1 &  0 &  0 & 0  \\  0 &  1 &    1&1 & 0  &0 \\  0 & 0  &0&  1  & 1  &0 \\    0 &  0  &  0 &0 & 1&1   \ea  \right ]   \]
is a totally non-negative matrix.

\begin{theorem}
In the factorization (\ref{fa}), assume each entry $d_{i,j}, j>i $  in  $D$ is non-negative.  $L$ is a $(0,1)$-matrix of double echelon form and doesn't contain  the sub-matrix (\ref{ta}). Then $A$ is a TNNP. 
\end{theorem}
\begin{proof}
From (\ref{ex}),  each term is non-negative, Then $a_{i,j}$ is non-negative for each $i,j$. Also, if one considers any sub-pfaffian  of $A$, then \cite{is}
\be 
Pf(A_{\alpha})= \sum_{K  \subset [2n], |K|=|\alpha|} Pf ( D_K)  \mathrm {det} (L_{ \alpha, K}),  \label{su}
\ee 
where $ |\alpha|   \leq  2n$ and $ (L_{ K, \alpha}) $ is the matrix obtained from $L$ by choosing the  K-columns corresponding to the index K and the $\alpha$-rows corresponding to the index $\alpha$. Here we notice that the index $K$ is composed of successive integers due to the tridiagonal form. Since $Pf ( D_K)$ and 
$ \mathrm {det} (L_{ \alpha, K})$  are non-negative for any indices $\alpha$ and $K$, it's known that $A$ is a TNNP by (\ref{su}). 
\end{proof}

Next, we also consider the non-singular extension of TNNP from $2n \times 2n  $ to $ 2(n+1) \times 2(n+1) $ using this factorization $(\ref{fa})$. One extends the unit  lower-triangular matrix $L$ to the unit  lower-triangular matrix  $\hat L $ as follows:
\bea
  \hat L =\left  [\ba{lll} L   &\vec 0^T & \vec 0^T  \\  \vec 0  &  1 & 0 \\ \vec 0 &0 & 1   \ea  \right ],   \label{et}
\eea 
where $ \vec 0= ( 0,0, \cdots ,0)$ is the zero vector ; furthermore, $D$ extends to the tridiagonal  form 
\bea
  \hat D =\left  [ \ba{cc} D   & \ba {cc}   \vec 0^T  & \vec 0^T  \\      d_{2n, 2n+1} & 0    \ea   \\
	\ba {cc}   \vec 0  &   -d_{2n, 2n+1} \\  \vec 0 & 0    \ea 
	& \ba {cc}    0  &   d_{2n+1, 2n+2} \\   -d_{2n+1, 2n+2}    & 0    \ea  \ea  \right ]. \label{etc}
\eea 
Then we have the following 
\begin{theorem}
In the factorization (\ref{fa}), assume $A$ is a   $2n \times 2n  $  TNNP. If $ d_{2n, 2n+1} \geq 0 $ and $ d_{ 2n+1,2n+2} \geq  0$, then $\hat A=\hat L \hat D \hat L^T$ defined in (\ref{et}) and  (\ref{etc}) is a $2(n+1) \times 2(n +1) $ TNNP. 
\end{theorem}
\begin{proof}
We consider  only the sub-pfaffians   of $\hat A$ containing the  $(2n+1)$st  or  $(2n+2)$nd columns. From Theorem 3.1, it is known that  each entry in $L$ and $d_{i,j}, j>i $  in $D$ is non-negative. The equation (\ref{ex}) implies the $(2n+1)$st  and $(2n+2)$nd columns of $\hat A$ are zeros except $\hat a_{2n, 2n+1} \geq 0$ and $\hat a_{2n+1, 2n+2}  \geq 0 $. For any sub-pfaffian  of $\hat A$ containing the  $(2n+1)$st  or  $(2n+2)$nd columns, the form $\hat L$ in (\ref{et}) obtains that the non-zero sub-pfaffians  containing the  $(2n+1)$st  or  $(2n+2)$nd columns are  the indices $K$  composed of successive integers  using the formula (\ref{su}). Also, each  determinant in (\ref{su}) is 0 or 1  for the indices $K$  composed of successive integers containing   $(2n+1)$  or  $(2n+2)$. Since $ Pf ( D_K)  \geq 0$ for each index K, we see that non-zero sub-pfaffians  containing the  $(2n+1)$st  or  $(2n+2)$nd columns are non-negative. This completes the proof.

\end{proof}

\section{Concluding Remarks}
\indent 

In this article, one investigates the construction of TNNP from the perfect matching in graph theory and lattice path enumeration to obtain non-singular solitons solutions for the BKP  equation (\ref{bkp}).  Also, one considers the factorization of TNNP into the  product of a lower-triangular matrix and a skew-symmetric tridiagonal matrix, thus establishing the TNNP and its extension. In the modified BKP equation \cite {gc, hi, wl} and the modified Veselov-Novikov equation \cite{ba, it, qp},  a new  $\tau$-function is  obtained  via the Darboux transformation to make soliton solution, which could lead us to consider  the extension of  a given TNNP. The extension  from a given TNNP to a new TNNP will be an interesting issue. Finally, the integrable deformations of affine surfaces are  described via the Nizhnik-Veselov-Novikov equation \cite{kp} and then TNNP could be  used to investigate these deformations.  Such studies  are deferred to future publication.

\subsection*{Acknowledgments}
This work is supported in part by the National Science and Technology Council of Taiwan under
Grant No. NSC 114-2115-M-606-001.

\end{document}